\RequirePackage{fix-cm}
\documentclass[twocolumn,epjc3]{svjour3}  
\usepackage{rotating}
\usepackage{xcolor}
\usepackage{amsmath}
\usepackage{hyphenat}
\smartqed  % flush right qed marks, e.g. at end of proof
\usepackage[english]{babel}
\usepackage{microtype}
\usepackage{tikz}
\usetikzlibrary{arrows.meta}
\usetikzlibrary{positioning,calc}
\RequirePackage{graphicx}

\usepackage[utf8]{inputenc}
\usepackage{amsfonts}
\usepackage{amssymb}
\usepackage{cite}
\usepackage{float}
\usepackage{hyperref}
\usepackage{booktabs}
\usepackage[normalem]{ulem}
\usepackage{orcidlink}
\usepackage{bm}

\journalname{Eur. Phys. J. C}

\begin{document} 

\title{Constitutive Geometry of the Effective Vacuum Sector of
  Pseudo-Complex General Relativity}



\author{Fridolin Weber\thanksref{e1,addr1,addr2}\orcidlink{0000-0002-5020-1906}
        \and
        Peter O. Hess\thanksref{e2,addr3,addr4}\orcidlink{0000-0002-6228-7348} \and 
        Cesar A. Zen Vasconcellos\thanksref{e3,addr5,addr6}\orcidlink{0000-0002-9814-5317}
}

\thankstext{e1}{e-mail: fweber@sdsu.edu}
\thankstext{e2}{e-mail: hess@nucleares.unam.mx}
\thankstext{e3}{e-mail: cesaraugustozenvasconcellos@gmail.com}

\institute{Department of Physics, San Diego State University, San Diego, CA, 92182, USA \label{addr1}
           \and
           Department of Physics, University of California San Diego, 
La Jolla, CA, 92093, USA \label{addr2}
           \and
 Instituto de Ciencias Nucleares, Universidad Nacional Autónoma de México, Mexico City, 04510, Mexico \label{addr3}
           \and
Frankfurt Institute for Advanced Studies (FIAS), 
J.W. von Goethe Universität, Frankfurt am Main, Germany \label{addr4}
           \and
          International Center for Relativistic Astrophysics Network (ICRANet), 
Pescara, Italy \label{addr5}
           \and
   Universidade Federal do Rio Grande do Sul, Porto Alegre, RS, 91501-970, Brazil \label{addr6}       
}

\date{Received: date / Accepted: date}

\maketitle

\begin{abstract} Effective vacuum sectors in modified theories of gravity
  are often represented by anisotropic stress-energy tensors, even
  when the underlying source is geometric rather than material. This
  raises a broader question: can such effective stresses be
  characterized intrinsically through a covariant constitutive
  structure, rather than modeled phenomenologically as exotic fluids?
  We address this question in pseudo-complex general relativity
  (pcGR), where a constrained variational formulation identifies the
  effective vacuum source as geometric in origin but does not uniquely
  determine the corresponding geometric constraint sector.

For the static pseudo-complex Schwarzschild solution, we show that the
effective vacuum tensor separates naturally into an isotropic
contribution and a purely transverse contribution associated with the
intrinsic geometry of the two-dimensional subspace orthogonal to the
timelike and radial directions. This decomposition leads to a
constitutive formulation based on the transverse metric. Its
infinitesimal deformations split uniquely into trace and traceless
shear modes, with
the static solution having only an isotropic transverse response.
 Within the class of local constitutive variables that depend
algebraically on the transverse metric and whose first variation
generates the trace projector with state-independent normalization,
the logarithmic transverse area strain is the minimal scalar
constitutive coordinate.

The corresponding local hyperelastic action reproduces the complete
algebraic structure of the static effective vacuum tensor. Matching to
the known pseudo-complex Schwarzschild solution then determines the
trace-sector constitutive potential uniquely within this constitutive
class and recovers the full radial dependence of the effective vacuum
density.  The rotating pseudo-complex Kerr solution retains the same
isotropic-plus-transverse tensorial decomposition but introduces
additional rotational dependence that cannot be captured by the static
trace-sector law alone,
thereby suggesting the traceless shear sector as the
natural geometric channel for extending the constitutive framework.

The resulting construction provides a minimal local constitutive
realization of the static effective vacuum sector rather than a
derivation of the full pseudo-complex theory. More generally, it shows
how an anisotropic effective vacuum source of geometric origin can be
organized as a constitutive response of an intrinsic transverse
geometry, providing a systematic framework for constraining and
classifying admissible geometric source sectors in extended
gravitational theories.
\end{abstract}

%%%%%%%%%%%%%%%%%%%%%%%%%%%%%%%%%%%%%%%%%%%%%%%%%%%%%%%%%%%%%%%%%%%%%%%%%%
\section{Introduction}
%%%%%%%%%%%%%%%%%%%%%%%%%%%%%%%%%%%%%%%%%%%%%%%%%%%%%%%%%%%%%%%%%%%%%%%%%%

Effective vacuum contributions arise in a broad class of extensions of
general relativity and in models of strong-field compact objects, where
they are frequently represented by anisotropic stress-energy tensors in
an effective Einstein equation. Although such a representation is
convenient, it can obscure an important distinction: the effective
source need not describe an independent material fluid. When the
additional stress is generated by the underlying geometric or
variational structure of the theory, the more fundamental question is
whether its anisotropy can be characterized directly in geometric
terms. In particular, one may ask whether an effective vacuum source can
be described as the constitutive response of an intrinsic geometric
sector, with its stress structure determined by geometric deformation
variables rather than by an independently prescribed exotic medium.

Pseudo-complex general relativity (pcGR) provides a particularly useful
setting in which to investigate this question. In pcGR, spacetime
coordinates and geometric quantities are promoted to pseudo-complex
variables, and the ordinary Einstein--Hilbert stationarity condition is
replaced by a weaker variational condition associated with the
zero-divisor structure of the pseudo-complex algebra
\cite{HessGreiner2009,HessSchaferGreiner2015}. After projection onto
ordinary spacetime, the field equations may be written in the effective
Einstein form
\begin{equation}
G_{\mu\nu} = 8\pi G \left( T_{\mu\nu}^{\rm matter}  + T_{\mu\nu}^{\rm
  pc} \right).
\label{eq:introEffectiveEinstein}
\end{equation}
The additional contribution \(T_{\mu\nu}^{\rm pc}\) modifies the
strong-field geometry of the known static and rotating compact-object
solutions while preserving the weak-field behavior of general
relativity. It is commonly represented as an anisotropic effective
stress-energy tensor and has been described phenomenologically as an
effective vacuum fluid, dark-energy-like source, or vacuum-polarization
contribution \cite{HessGreiner2009,Caspar2012}. Such descriptions are
useful at the level of the projected field equations, but they do not by
themselves identify the geometric structure responsible for the
effective source.

A constrained variational formulation of pcGR provides a more precise
starting point. In that framework, the pseudo-complex variational
condition is interpreted as projected stationarity, and the effective
source arises as the variational response of a covariant constraint
sector relating the two projected metric components. The effective
stress tensor is therefore geometric in origin rather than an
independently introduced matter contribution. The constrained
variational framework does not, however, uniquely determine the
functional form of the constraint sector. The problem of selecting and
organizing physically admissible realizations of this geometric sector
consequently remains open.

The purpose of the present work is to determine whether this
undetermined geometric sector admits an intrinsic constitutive
description. We do not assume that spacetime itself is a physical
elastic medium. Instead, we ask whether the projected effective vacuum
source can be organized as the constitutive response of its transverse
geometry. To address this question, we use the mathematical structure
of classical hyperelastic constitutive theory and its geometric
formulations
\cite{TruesdellNoll2004,Ogden1984,MarsdenHughes1994,Holzapfel2000},
but apply it entirely to the effective geometric sector of pcGR. The
resulting constitutive variables characterize deformations of the
transverse geometry and do not introduce an additional material medium
or matter degree of freedom.

The motivation comes from the algebraic structure of the static
effective vacuum tensor. For the pseudo-complex Schwarzschild solution
\cite{Caspar2012}, the radial equation of state satisfies
\begin{equation}
p_r=-\rho,
\label{eq:introRadialEOS}
\end{equation}
which permits the effective source to be rewritten as
\begin{equation}
T_{\mu\nu}^{\rm pc} = -\rho\,g_{\mu\nu} + \left( p_\perp-p_r \right)
q_{\mu\nu},
\label{eq:introDecomposition}
\end{equation}
where \(q_{\mu\nu}\) is the projector onto the two-dimensional subspace
orthogonal to the static timelike and radial directions. Equation
(\ref{eq:introDecomposition}) separates the source into an isotropic
contribution proportional to the spacetime metric and a purely
transverse response. This decomposition suggests that the nontrivial
constitutive information is carried by the geometry of the transverse
two-space.

Infinitesimal deformations of the transverse metric admit a natural
irreducible decomposition.
Every infinitesimal transverse deformation splits uniquely into a trace
part and a traceless part,
\begin{equation}
\delta q_{AB} = \frac{1}{2} q_{AB} q^{CD} \delta q_{CD} + \delta
q_{AB}^{\rm TF}.
\label{eq:introTraceShear}
\end{equation}
The trace mode changes the local transverse area, whereas the traceless
mode describes shear at fixed area. Static spherical symmetry permits
only an isotropic transverse response at first order and therefore
selects the trace channel. The scalar coordinate naturally associated
with this channel is the logarithmic transverse area strain,
\begin{equation}
\zeta = - \ln \left( \frac{\sqrt q} {\sqrt{q_0}} \right),
\label{eq:introZeta}
\end{equation}
whose metric variation is
\begin{equation}
\delta\zeta = \frac{1}{2} q_{\mu\nu} \delta g^{\mu\nu}.
\label{eq:introZetaVariation}
\end{equation}
A derivation of this identity is given in \ref{app:geometry}.
The transverse projector appearing in the effective source is
therefore generated directly by the metric variation of the
logarithmic area strain.

Under the assumptions that the constitutive sector is local, contains
no derivatives of the transverse metric, is transversely isotropic, and
depends on a single scalar trace variable, the minimal action takes the
form
\begin{equation}
I_{\rm con} = - \int d^4x\,\sqrt{-g}\,W(\zeta).
\label{eq:introConstitutiveAction}
\end{equation}

Because the constitutive response is derived from a scalar
strain-energy density, the resulting constitutive framework is
hyperelastic in the terminology of continuum mechanics
\cite{Ogden1984,Holzapfel2000}. Throughout this paper, however, the
term hyperelastic refers solely to the mathematical structure of the
constitutive model and does not imply the existence of a physical
material medium.

Under the restricted metric variation described below,
Eq.~\eqref{eq:introConstitutiveAction} gives
\begin{equation}
T^{\rm con}_{\mu\nu} = -W(\zeta)g_{\mu\nu} +
W_{,\zeta}(\zeta)q_{\mu\nu}.
\label{eq:introConstitutiveTensor}
\end{equation}
Comparison of Eqs.~(\ref{eq:introConstitutiveTensor}) and
(\ref{eq:introDecomposition}) establishes the constitutive
correspondence
\begin{equation}
W(\zeta)=\rho, \qquad W_{,\zeta}(\zeta)=p_\perp-p_r.
\label{eq:introCorrespondence}
\end{equation}
For the static pseudo-complex Schwarzschild solution, this
correspondence uniquely determines the associated trace-sector
constitutive potential,
\begin{equation}
W(\zeta) = W_0 \exp\!\left(\frac{n+2}{2}\zeta\right).
\label{eq:introExponentialPotential}
\end{equation}
The derivation of this result is presented in
Section~\ref{sec:static-constitutive-law} as an inverse constitutive
problem. The resulting exponential potential reproduces the complete
radial dependence of the static effective vacuum density.

The resulting construction should not be interpreted as a unique
derivation of the complete pseudo-complex theory from elasticity.
The original pseudo-complex variational principle determines the
projected variational structure but does not uniquely select a
transverse constraint tensor or exclude additional off-shell
invariants.
Likewise,
the static solution determines the constitutive potential only along the
isotropic trace sector. Terms involving transverse shear, curvature, or
other invariants may remain invisible on the exact spherical
background.

The claim made here is therefore more specific. We show that the
static effective vacuum sector admits a minimal local hyperelastic
realization on the transverse two-space and that, within the stated
constitutive class, the logarithmic area strain is the natural scalar
variable conjugate to the transverse isotropic stress.  The rotating
pseudo-complex Kerr solution then provides evidence that additional
transverse constitutive information becomes relevant once spherical
symmetry is broken, with the traceless shear sector providing the
natural geometric extension.  In this sense, the constitutive
formulation offers a systematic language for organizing the otherwise
undetermined geometric constraint sector of pcGR.

The paper is organized as follows. Section~\ref{sec:effective-vacuum}
reviews the effective vacuum tensor of the static pseudo-complex
Schwarzschild solution and derives its decomposition into isotropic and
transverse contributions. Section~\ref{sec:transverse-geometry}
develops the intrinsic geometry of the transverse two-space and its
irreducible decomposition into trace and shear modes. The corresponding
minimal local constitutive action and its associated stress tensor are
derived in Section~\ref{sec:minimal-action}. In
Section~\ref{sec:static-constitutive-law}, the known static pcGR  
solution is used to reconstruct the trace-sector constitutive potential.
Section~\ref{sec:rotation} discusses the rotating extension and the role
of transverse shear. Section~\ref{sec:constraint-selection} relates the
constitutive formulation to the constrained variational framework and
examines its implications for selecting admissible geometric constraint
sectors. The final section summarizes the principal results and outlines
directions for future work.

Three appendices provide supporting mathematical details. \ref{app:geometry}
derives the geometric identities relating variations of the transverse
metric to the logarithmic transverse area strain. \ref{app:representation} presents
the representation-theoretic decomposition of transverse metric
deformations into trace and shear components. \ref{app:minimality} proves the
minimality theorem for the logarithmic transverse area strain within the
class of local constitutive theories considered in this work.

%%%%%%%%%%%%%%%%%%%%%%%%%%%%%%%%%%%%%%%%%%%%%%%%%%%%%%%%%%%%%%%%%%%%%%%%%%
\section{Effective Vacuum Sector of Pseudo-Complex General Relativity}
\label{sec:effective-vacuum}
%%%%%%%%%%%%%%%%%%%%%%%%%%%%%%%%%%%%%%%%%%%%%%%%%%%%%%%%%%%%%%%%%%%%%%%%%%

The purpose of this section is to identify the geometric structure of the
effective vacuum sector that any constitutive description must reproduce.
Throughout this work we restrict attention initially to the static,
spherically symmetric pseudo-complex Schwarzschild solution \cite{Caspar2012}. The rotating
case will be discussed in Section~\ref{sec:rotation}.

The projected field equations of pseudo-complex general relativity may be
written in the effective Einstein form \cite{HessGreiner2009,Caspar2012}
\begin{equation}
G_{\mu\nu} = 8\pi G \left( T_{\mu\nu}^{\rm matter} + T_{\mu\nu}^{\rm
  pc} \right),
\label{eq:effectiveEinstein}
\end{equation}
where the tensor \(T_{\mu\nu}^{\rm pc}\) represents the effective
geometric source generated by the pseudo-complex extension.
For the static Schwarzschild solution the spacetime metric is \cite{Caspar2012}
\begin{equation}
ds^2 = -f(r)\,dt^2 + \frac{dr^2}{f(r)} + r^2d\Omega^2,
\label{eq:metric}
\end{equation}
with
\begin{equation}
f(r) = 1 - \frac{2M}{r} + \frac{C}{r^n},
\label{eq:fr}
\end{equation}
where
\begin{equation}
C = \frac{B_n} {(n-1)(n-2)}.
\end{equation}

The effective vacuum tensor may be expressed in the standard anisotropic
form
\begin{equation}
T_{\mu\nu}^{\rm pc} = \rho\,u_\mu u_\nu + p_r\,s_\mu s_\nu +
p_\perp\,q_{\mu\nu},
\label{eq:anisotropic}
\end{equation}
where
\begin{equation}
u_\mu u^\mu=-1, \qquad s_\mu s^\mu=1, \qquad u_\mu s^\mu=0,
\end{equation}
and
\begin{equation}
q_{\mu\nu} = g_{\mu\nu} + u_\mu u_\nu - s_\mu s_\nu
\label{eq:qprojector}
\end{equation}
projects onto the two-dimensional subspace orthogonal to the timelike
and radial directions. The projector satisfies
\begin{equation}
q_{\mu\nu}u^\nu=0, \qquad q_{\mu\nu}s^\nu=0, \qquad q^\mu{}_\mu=2,
\end{equation}
and therefore defines the intrinsic geometry of the transverse
two-space.

%%%%%%%%%%%%%%%%%%%%%%%%%%%%%%%%%%%%%%%%%%%%%%%%%%%%%%%%%%%%%%

\subsection{Static effective equation of state}

Evaluation of the Einstein tensor gives
\begin{equation}
\rho = \frac{B_n} {8\pi G(n-2)r^{n+2}},
\label{eq:rho}
\end{equation}
together with
\begin{equation}
p_r = -\rho,
\label{eq:pr}
\end{equation}
and
\begin{equation}
p_\perp - p_r = \frac{(n+2)B_n} {16\pi G(n-2)r^{n+2}}.
\label{eq:deltap}
\end{equation}
Consequently,
\begin{equation}
  \frac{p_\perp - p_r}{\rho} = \frac{n+2}{2}.
  \label{eq:staticRatio}
\end{equation}

Equation~(\ref{eq:staticRatio}) is remarkable because the ratio of
transverse anisotropic stress to effective vacuum energy density is
constant throughout the static vacuum solution. As will be shown in
Section~\ref{sec:static-constitutive-law}, this immediately determines the
constitutive law of the minimal hyperelastic realization.

%%%%%%%%%%%%%%%%%%%%%%%%%%%%%%%%%%%%%%%%%%%%%%%%%%%%%%%%%%%%%%

\subsection{Geometric decomposition}

Using the equation of state (\ref{eq:pr}), the radial pressure may be
eliminated, $ p_r=-\rho$, so that Eq.~(\ref{eq:anisotropic}) becomes
\begin{equation}
T_{\mu\nu}^{\rm pc} = \rho \left( u_\mu u_\nu - s_\mu s_\nu \right) +
p_\perp q_{\mu\nu}.
\end{equation}
Using the identity
\begin{equation}
g_{\mu\nu} = - u_\mu u_\nu + s_\mu s_\nu + q_{\mu\nu},
\end{equation}
one finds
\begin{equation}
u_\mu u_\nu - s_\mu s_\nu = - g_{\mu\nu} + q_{\mu\nu},
\end{equation}
which immediately yields
\begin{equation}
T_{\mu\nu}^{\rm pc} = - \rho g_{\mu\nu} + (p_\perp-p_r) q_{\mu\nu}.
\label{eq:decomposition}
\end{equation}
Equation (\ref{eq:decomposition}) is the central structural observation
of the present work.

%%%%%%%%%%%%%%%%%%%%%%%%%%%%%%%%%%%%%%%%%%%%%%%%%%%%%%%%%%%%%%

\subsection{Constitutive interpretation}

Equation (\ref{eq:decomposition}) separates the effective source into two
geometrically distinct pieces.
The first contribution is isotropic and proportional to the spacetime
metric.
The second acts exclusively on the transverse geometry through the
projector \(q_{\mu\nu}\).
This separation has an important consequence.
The anisotropic part of the effective vacuum sector is entirely
encoded in the transverse two-space.
No independent radial
stress degree of freedom survives after imposing the effective
equation of state (\ref{eq:pr}).   Rather than introducing
Eq.~(\ref{eq:decomposition}) as the stress tensor of an independent
anisotropic vacuum fluid, we shall investigate whether it can be
represented as the constitutive response of the transverse geometry.
The remainder of the paper develops this idea systematically.
First, the intrinsic geometry of the transverse two-space is analysed
independently of the pseudo-complex field equations. We then show that
its trace deformation naturally generates the projector
\(q_{\mu\nu}\), leading to a minimal local constitutive action whose
variation reproduces Eq.~(\ref{eq:decomposition}).

%%%%%%%%%%%%%%%%%%%%%%%%%%%%%%%%%%%%%%%%%%%%%%%%%%%%%%%%%%%%%%%%%%%%%%%%%%
\section{Geometry of the Transverse Two-Space}
\label{sec:transverse-geometry}
%%%%%%%%%%%%%%%%%%%%%%%%%%%%%%%%%%%%%%%%%%%%%%%%%%%%%%%%%%%%%%%%%%%%%%%%%%

The decomposition obtained in the previous section shows that the
anisotropic part of the effective vacuum sector acts exclusively on the
two-dimensional subspace orthogonal to the timelike and radial
directions. The purpose of the present section is therefore to develop
the intrinsic geometry of this transverse space independently of any
constitutive assumptions.

Rather than introducing elasticity at the outset, we first identify the
natural kinematic variables associated with the transverse geometry.
Hyperelasticity will emerge only after these geometric structures have
been established.

%%%%%%%%%%%%%%%%%%%%%%%%%%%%%%%%%%%%%%%%%%%%%%%%%%%%%%%%%%%%%%

\subsection{Transverse geometry}

The effective vacuum tensor introduced in the previous section singles
out the two-dimensional subspace orthogonal to the timelike and radial
directions. This transverse geometry provides the natural kinematic
setting for the constitutive description developed in the remainder of
this work.

Let $u^\mu$ denote the unit timelike vector field and $s^\mu$ the unit
spacelike radial vector field, satisfying
\begin{equation}
u^\mu u_\mu=-1,\qquad
s^\mu s_\mu=1,\qquad
u^\mu s_\mu=0.
\end{equation}
The projector onto the orthogonal two-dimensional subspace is then
\begin{equation}
q_{\mu\nu}
=
g_{\mu\nu}
+
u_\mu u_\nu
-
s_\mu s_\nu.
\end{equation}
This projector induces the intrinsic metric $q_{AB}$ on the transverse
manifold, whose determinant is denoted by
\begin{equation}
q=\det(q_{AB}).
\end{equation}
The intrinsic metric $q_{AB}$ completely characterizes the geometry of
the transverse two-space and serves as the fundamental configuration
variable of the constitutive theory. Metric-based formulations of
elasticity and surface elasticity provide the corresponding
mathematical precedent \cite{MarsdenHughes1994,GurtinMurdoch1975,
  EfratiSharonKupferman2009,Yavari2026}.  In the present construction,
constitutive quantities are built from local deformations of the
transverse metric. Unlike ordinary elasticity, however, no physical
material medium or material surface is introduced.
 The constitutive description is
entirely geometric and refers only to the effective vacuum sector of
pseudo-complex gravity.

Since the transverse manifold is two-dimensional, every infinitesimal
variation of its metric admits a unique decomposition into trace and
traceless parts. These correspond respectively to local changes of
transverse area and local shear deformations that preserve the area. The
derivation of this decomposition, together with its
representation-theoretic interpretation under $\mathrm{SO}(2)$, is
summarized in \ref{app:representation}.

Accordingly, the transverse metric defines the natural constitutive
configuration space of the effective vacuum sector. The following
subsections identify its irreducible kinematic modes and show that, for
the static pseudo-complex Schwarzschild solution, only the trace degree
of freedom enters the constitutive response.

%%%%%%%%%%%%%%%%%%%%%%%%%%%%%%%%%%%%%%%%%%%%%%%%%%%%%%%%%%%%%%

\subsection{Irreducible decomposition}

At each point of the transverse two-dimensional manifold, the tangent
space carries the natural action of the rotation group $SO(2)$.
Consequently, every infinitesimal deformation of the transverse metric
admits a unique decomposition into irreducible $SO(2)$ components,
\begin{equation}
\delta q_{AB} = \delta q^{\rm tr}_{AB} + \delta q^{\rm TF}_{AB},
\label{eq:split}
\end{equation}
where the trace component is
\begin{equation}
\delta q^{\rm tr}_{AB}
=
\frac12 q_{AB} q^{CD}\delta q_{CD},
\end{equation}
and the traceless component is
\begin{equation}
\delta q^{\rm TF}_{AB} = \delta q_{AB} - \frac12 q_{AB} q^{CD}\delta
q_{CD}.
\end{equation}
By construction,
\begin{equation}
q^{AB}\delta q^{\rm TF}_{AB}=0.
\end{equation}

The trace component describes local changes of the transverse area,
whereas the traceless component describes shear deformations that
preserve the area to first order. Representation-theoretically, the
space of symmetric rank-two tensors therefore decomposes as
\begin{equation}
\mathrm{Sym}^2(T^*\Sigma) = \mathbf{1} \oplus \mathbf{2},
\end{equation}
where $\mathbf{1}$ denotes the one-dimensional trace representation and
$\mathbf{2}$ the two-dimensional traceless symmetric representation of
$SO(2)$.

This decomposition is a purely geometric property of symmetric tensors
on a two-dimensional Riemannian manifold and is therefore independent
of the pseudo-complex field equations. The dynamics determines only
which of these geometric modes contribute to a given solution. As
shown below, the static pseudo-complex Schwarzschild solution excites
only the trace sector. In the rotating case, the trace-sector response
acquires dependence on additional rotational geometry.
Transverse shear invariants provide one natural possible extension of
the trace-sector potential, as occurs generally in constitutive
theories depending on the full surface metric
\cite{EfratiSharonKupferman2009,Yavari2026}.
The known rotating stress tensor, however, does not uniquely identify
such invariants \cite{Caspar2012}.

%%%%%%%%%%%%%%%%%%%%%%%%%%%%%%%%%%%%%%%%%%%%%%%%%%%%%%%%%%%%%%

\subsection{Trace deformation and transverse area}

The decomposition of the previous subsection immediately identifies the
trace mode as the unique infinitesimal deformation capable of changing
the intrinsic area of the transverse geometry. This follows from the
variation of the determinant of the transverse metric.

Using the standard determinant identity,
\begin{equation}
\delta\ln q = q^{AB}\delta q_{AB},
\end{equation}
one obtains
\begin{equation}
\delta\ln\sqrt{q} = \frac12 q^{AB}\delta q_{AB}.
\label{eq:dlogsqrtq}
\end{equation}
Thus the fractional variation of the transverse area element depends
only on the trace of the metric deformation.

By contrast, the traceless component satisfies
\begin{equation}
q^{AB}\delta q^{\rm TF}_{AB}=0,
\end{equation}
and therefore leaves the transverse area invariant to first order.
Consequently, infinitesimal deformations separate naturally into two
kinematically distinct classes: trace deformations, which modify the
local transverse area, and shear deformations, which preserve it.

This distinction is purely geometric and follows solely from the
two-dimensional structure of the transverse manifold. It therefore
precedes any constitutive assumptions or field equations. Since the
static pseudo-complex Schwarzschild solution excites only the trace
sector, the transverse area provides the natural scalar quantity from
which the constitutive variable will be constructed in the following
subsection.

%%%%%%%%%%%%%%%%%%%%%%%%%%%%%%%%%%%%%%%%%%%%%%%%%%%%%%%%%%%%%%

\subsection{Logarithmic transverse area strain}

Within the present two-dimensional, algebraic setting, the transverse
area provides the natural scalar associated with the trace deformation.
Accordingly, the constitutive variable describing the static effective
vacuum sector is naturally constructed from the transverse area element.

A natural dimensionless measure is the logarithmic transverse area
strain, analogous to the logarithmic strain used in finite
elasticity and to its geometric interpretation on deformation space
\cite{Hencky1928,NeffEidelOsterbrinkMartin2014},
\begin{equation}
\zeta = - \ln\!\left( \frac{\sqrt{q}} {\sqrt{q_0}} \right),
\label{eq:zeta}
\end{equation}
where $q_0$ denotes the determinant of the transverse metric in a
reference configuration. The logarithmic form measures fractional area
changes and vanishes in the reference state.

Its first variation follows directly from
Eq.~(\ref{eq:dlogsqrtq}),
\begin{equation}
\delta\zeta = - \frac12 q^{AB}\delta q_{AB}.
\label{eq:deltazeta}
\end{equation}
Expressing the variation in terms of the spacetime metric using
\begin{equation}
q^{AB}\delta q_{AB} = - q_{\mu\nu}\delta g^{\mu\nu},
\end{equation}
one obtains
\begin{equation}
\delta\zeta = \frac12 q_{\mu\nu}\delta g^{\mu\nu}.
\label{eq:spacetimevariation}
\end{equation}

Equation~(\ref{eq:spacetimevariation}) shows that the first variation of
the logarithmic transverse area strain is generated precisely by the
transverse projector appearing in the effective pseudo-complex vacuum
tensor. This establishes the direct geometric link between the
constitutive variable and the algebraic structure of the effective
vacuum sector.

As shown in \ref{app:minimality}, the logarithmic transverse
area strain is the minimal local constitutive coordinate whose first
variation generates the transverse trace projector with a
state-independent normalization.  It therefore provides the natural
constitutive variable for the minimal local constitutive theory
developed in the following section.

%%%%%%%%%%%%%%%%%%%%%%%%%%%%%%%%%%%%%%%%%%%%%%%%%%%%%%%%%%%%%%

\subsection{Minimal constitutive variable}

The logarithmic transverse strain introduced above is not merely a
convenient parametrization of the static geometry. Rather, it is singled
out by the constitutive structure of the effective vacuum sector.

Consider the class of local constitutive theories whose scalar
constitutive variable
\begin{enumerate}
\item depends locally on the transverse metric,
\item contains no derivatives of the transverse metric, and
\item couples only to the isotropic trace deformation with a
state-independent normalization.
\end{enumerate}
Within this class, the logarithmic transverse strain is the minimal
local scalar whose first variation is proportional to the trace
deformation $q^{AB}\delta q_{AB}$ of the transverse metric.
Consequently, up to normalization and an additive constant, it is the
unique local constitutive coordinate within the class defined above
whose first variation generates the transverse trace projector with a
state-independent normalization.

The proof of this minimality result is given in
\ref{app:minimality}. The argument establishes minimality only
within the constitutive class specified above. It does not imply that
more general constitutive theories are excluded.

Indeed, additional local or nonlocal constitutive invariants, such as
transverse shear or curvature invariants, may be incorporated in more
general constitutive models. Such additional contributions are not
determined by the exact static spherically symmetric solution and
therefore do not affect the minimal trace-sector reconstruction
developed here. They may become relevant when extending the theory
beyond the static trace sector, as discussed later for the rotating
pseudo-complex Kerr geometry.

%%%%%%%%%%%%%%%%%%%%%%%%%%%%%%%%%%%%%%%%%%%%%%%%%%%%%%%%%%%%%%

\subsection{Geometric interpretation}

The preceding analysis identifies two independent kinematic modes of the
transverse geometry. The first is the trace mode, represented by the
logarithmic transverse strain $\zeta$, which measures local changes of
the transverse area. The second is the traceless tensor
$\delta q^{\rm TF}_{AB}$, which describes transverse shear while
preserving the area to first order. Together these constitute the
elementary deformation modes of the transverse constitutive
configuration space.

As shown in \ref{app:representation}, these modes correspond to
the two irreducible representations of the local rotation group
$SO(2)$ acting on the transverse tangent space. The trace mode forms the
one-dimensional scalar representation, whereas the traceless symmetric
tensor forms the two-dimensional shear representation. This
representation-theoretic structure is entirely geometric and is
independent of the pseudo-complex field equations.

The effective vacuum solutions determine which of these geometric
modes are realized. For the static pseudo-complex Schwarzschild
solution, only the trace representation is excited, so the
constitutive response is completely described by the logarithmic
transverse strain.

The decomposition developed in this section therefore provides the kinematic
foundation of the constitutive theory. The following section constructs the
minimal local constitutive action for the static trace sector, while the
rotating extension is considered in Section~\ref{sec:rotation}.

%%%%%%%%%%%%%%%%%%%%%%%%%%%%%%%%%%%%%%%%%%%%%%%%%%%%%%%%%%%%%%%%%%%%%%%%%%
\section{Minimal Local Constitutive Action}
\label{sec:minimal-action}
%%%%%%%%%%%%%%%%%%%%%%%%%%%%%%%%%%%%%%%%%%%%%%%%%%%%%%%%%%%%%%%%%%%%%%%%%%

The previous section established that the transverse geometry possesses
two independent kinematic modes: a trace mode associated with local area
change and a traceless shear mode describing shape deformation at fixed
area. It also identified the logarithmic transverse strain as the
minimal local constitutive variable for the static trace sector.

For the exact static spherically symmetric background, only the trace
mode contributes to the effective vacuum tensor. The objective of the
present section is therefore to construct the minimal local
constitutive action compatible with this geometric and kinematic
structure. The resulting action provides the simplest constitutive
realization of the effective vacuum sector while remaining entirely
within the class of local constitutive theories identified in the
previous section.

%%%%%%%%%%%%%%%%%%%%%%%%%%%%%%%%%%%%%%%%%%%%%%%%%%%%%%%%%%%%%%

\subsection{Constitutive assumptions}

The construction is based on the following constitutive assumptions.

\begin{enumerate}
\item The constitutive description is local.
\item The constitutive action depends only on the local transverse
geometry and not on its derivatives.
\item The static background is transversely isotropic.
\item Only the trace deformation contributes to the constitutive
response.
\end{enumerate}

Under these assumptions, the constitutive state is completely
determined by the logarithmic transverse strain introduced in
Eq.~(\ref{eq:zeta}), whose minimality within the present constitutive
class is proved in \ref{app:minimality}.

These assumptions define the minimal constitutive model appropriate to
the static pseudo-complex Schwarzschild solution. They do not exclude
more general constitutive theories containing additional local or
nonlocal invariants, such as transverse shear or curvature terms, which
become relevant beyond the static trace sector considered here.

%%%%%%%%%%%%%%%%%%%%%%%%%%%%%%%%%%%%%%%%%%%%%%%%%%%%%%%%%%%%%%

\subsection{Constitutive action}

The constitutive assumptions formulated above restrict the admissible
action to depend only on the local logarithmic transverse strain.
In hyperelastic constitutive theory, the stress response is generated
from a scalar stored-energy function of the chosen strain variables
\cite{TruesdellNoll2004,Ogden1984,Holzapfel2000,Yavari2026}.
Within the present constitutive class, the corresponding minimal local
action for the geometric sector is
\begin{equation}
I_{\rm con} = - \int d^4x\,\sqrt{-g}\,W(\zeta),
\label{eq:constitutiveAction}
\end{equation}
where $W(\zeta)$ is an arbitrary constitutive potential.
At this stage, no specific constitutive law has been imposed.
Equation~(\ref{eq:constitutiveAction}) therefore represents the most general local
trace-sector constitutive realization of the static effective vacuum
sector within the constitutive class defined by the assumptions of the
previous subsection. All
constitutive freedom is contained in the choice of the scalar potential
$W(\zeta)$.
The pseudo-complex field equations do not determine the functional form
of $W$ directly within the present constitutive construction. Instead,
$W$ is reconstructed by requiring that the stress tensor derived from
the constitutive action reproduce the effective vacuum sector of the
known static pseudo-complex Schwarzschild solution. This inverse
constitutive problem is carried out in the following section.

%%%%%%%%%%%%%%%%%%%%%%%%%%%%%%%%%%%%%%%%%%%%%%%%%%%%%%%%%%%%%%

\subsection{Variation}

The constitutive response is obtained by varying the
action (\ref{eq:constitutiveAction}) with respect to the spacetime
metric, subject to the restriction that the timelike and radial
directions defining the transverse splitting are held fixed.
Under this restricted metric variation, the dependence of the
logarithmic transverse strain on the metric is through the induced
transverse geometry, as given by Eq.~(\ref{eq:spacetimevariation}).
Using the standard identity
\begin{equation}
\delta\sqrt{-g} = -\frac{1}{2}\sqrt{-g}\,g_{\mu\nu}\,\delta
g^{\mu\nu},
\end{equation}
together with Eq.~(\ref{eq:spacetimevariation}), one finds
\begin{equation}
\delta I_{\rm con} = \frac{1}{2} \int d^4x\,\sqrt{-g} \left[ W
  g_{\mu\nu} - W_{,\zeta}q_{\mu\nu} \right] \delta g^{\mu\nu}.
\label{eq:actionvariation}
\end{equation}
Defining the constitutive stress tensor by
\begin{equation}
\delta I_{\rm con} = -\frac{1}{2} \int d^4x\,\sqrt{-g}\, T^{\rm
  con}_{\mu\nu}\,\delta g^{\mu\nu},
\end{equation}
then gives
\begin{equation}
T^{\rm con}_{\mu\nu} = -W(\zeta)g_{\mu\nu} + W_{,\zeta}q_{\mu\nu}.
\label{eq:Tcon}
\end{equation}
Equation~(\ref{eq:Tcon}) should therefore be understood as the
covariant representative of the restricted transverse constitutive
response; it is not intended as a derivation of the most general
four-dimensional metric stress tensor.

The variation separates naturally into an isotropic contribution
proportional to the spacetime metric and a purely transverse
contribution proportional to the transverse projector. This is
precisely the tensorial structure identified in the effective
pseudo-complex vacuum sector. The constitutive potential determines
only the scalar coefficients multiplying these two geometric tensors.

%%%%%%%%%%%%%%%%%%%%%%%%%%%%%%%%%%%%%%%%%%%%%%%%%%%%%%%%%%%%%%

\subsection{Constitutive stress tensor}

Under the restricted metric variation introduced above, the
constitutive stress tensor is represented covariantly by
\begin{equation}
\delta I_{\rm con} = -\frac{1}{2} \int d^4x\,\sqrt{-g}\, T^{\rm
  con}_{\mu\nu}\,\delta g^{\mu\nu}.
\label{eq:stressdefinition}
\end{equation}
This construction is the geometric analogue of the hyperelastic
stress obtained from the variation of a stored-energy density with
respect to the relevant configuration variables
\cite{Ogden1984,Holzapfel2000,Yavari2026}.
Comparison with
Eq.~(\ref{eq:actionvariation}) immediately yields
\begin{equation}
T^{\rm con}_{\mu\nu} = - W(\zeta)\,g_{\mu\nu} +
W_{,\zeta}\,q_{\mu\nu},
\label{eq:constitutiveTensor}
\end{equation}
where the constitutive potential determines the isotropic contribution,
while its derivative with respect to the logarithmic transverse strain
determines the transverse anisotropic response.

Equation~(\ref{eq:constitutiveTensor}) is the fundamental constitutive
result of the present work. It shows that the constitutive stress tensor
possesses exactly the same tensorial decomposition as the effective
vacuum tensor of pseudo-complex general relativity: an isotropic term
proportional to the spacetime metric and a purely transverse term
proportional to the transverse projector. At this stage, however, the
constitutive potential remains arbitrary. The following section
determines its functional form by matching the constitutive stress
tensor to the known effective vacuum solution of the static
pseudo-complex Schwarzschild geometry.

%%%%%%%%%%%%%%%%%%%%%%%%%%%%%%%%%%%%%%%%%%%%%%%%%%%%%%%%%%%%%%

\subsection{Matching to the effective vacuum sector}

The constitutive stress tensor derived above possesses exactly the same
tensorial decomposition as the effective vacuum tensor obtained from the
projected pseudo-complex field equations. Recalling the result of
Section~\ref{sec:effective-vacuum},
\begin{equation}
T_{\mu\nu}^{\rm pc}
=
-
\rho\,g_{\mu\nu}
+
(p_\perp-p_r)\,q_{\mu\nu},
\end{equation}
comparison with Eq.~(\ref{eq:constitutiveTensor}) immediately yields the
constitutive identifications
\begin{equation}
W(\zeta) = \rho,
\end{equation}
and
\begin{equation}
W_{,\zeta} = p_\perp-p_r.
\label{eq:constitutiveMatch}
\end{equation}

These relations establish the constitutive correspondence between the
minimal local constitutive theory and the effective vacuum sector of
pseudo-complex gravity. The constitutive potential is identified with
the effective vacuum energy density, while its derivative with respect
to the logarithmic transverse strain determines the transverse
anisotropic response.

The tensorial matching is entirely algebraic and does not rely on the
explicit radial form of the static solution. The latter enters only in
the following section, where the constitutive potential is reconstructed
from the known effective vacuum profile.

The correspondence established above completes the construction of the
minimal local constitutive framework. At this stage, the tensorial
structure of the effective vacuum sector has been reproduced without
specifying a constitutive law. The remaining task is to determine the
constitutive potential by matching the constitutive stress tensor to the
known static pseudo-complex Schwarzschild solution.

%%%%%%%%%%%%%%%%%%%%%%%%%%%%%%%%%%%%%%%%%%%%%%%%%%%%%%%%%%%%%%%%%%%%%%%%%%
\section{Constitutive Law of the Static Vacuum Sector}
\label{sec:static-constitutive-law}
%%%%%%%%%%%%%%%%%%%%%%%%%%%%%%%%%%%%%%%%%%%%%%%%%%%%%%%%%%%%%%%%%%%%%%%%%%

The previous section established the minimal local constitutive
framework and showed that its stress tensor possesses exactly the same
tensorial structure as the effective vacuum sector of pseudo-complex
gravity. At that stage, however, the constitutive potential
$W(\zeta)$ remained undetermined.

The purpose of the present section is to determine this potential from
the known static pseudo-complex Schwarzschild solution. This constitutes
an inverse constitutive problem: rather than postulating a constitutive
law and deriving the stress tensor, the known effective vacuum sector is
used to reconstruct the constitutive potential.

%%%%%%%%%%%%%%%%%%%%%%%%%%%%%%%%%%%%%%%%%%%%%%%%%%%%%%%%%%%%%%

\subsection{Static effective vacuum}

For the static pseudo-complex Schwarzschild solution, the effective
vacuum sector is characterized by
\begin{equation}
\rho(r) = \frac{B_n} {8\pi G(n-2)r^{n+2}},
\label{eq:rhoagain}
\end{equation}
together with
\begin{equation}
p_r=-\rho,
\end{equation}
and
\begin{equation}
p_\perp-p_r = \frac{(n+2)B_n} {16\pi G(n-2)r^{n+2}}.
\end{equation}
It follows immediately that
\begin{equation}
\frac{p_\perp-p_r}{\rho} = \frac{n+2}{2}.
\label{eq:ratioagain}
\end{equation}

Although the individual stress components vary with the radial
coordinate, their ratio is constant. This property makes it possible to
reconstruct the constitutive potential independently of the explicit
radial dependence of the solution.

%%%%%%%%%%%%%%%%%%%%%%%%%%%%%%%%%%%%%%%%%%%%%%%%%%%%%%%%%%%%%%
\subsection{Inverse constitutive problem}

The constitutive correspondence established in
Section~\ref{sec:minimal-action} identifies the constitutive potential
with the effective vacuum energy density and its derivative with the
transverse anisotropic response,
\begin{equation}
W=\rho,
\qquad
W_{,\zeta}=p_\perp-p_r.
\end{equation}
Substituting these relations into
Eq.~(\ref{eq:ratioagain}) immediately yields
\begin{equation}
\frac{dW}{d\zeta} = \frac{n+2}{2}\,W.
\label{eq:constitutiveODE}
\end{equation}

Equation~(\ref{eq:constitutiveODE}) is the constitutive equation of the
static trace sector. It determines the constitutive potential directly
from the known effective vacuum solution without requiring any
independent constitutive ansatz. In this sense, the static
pseudo-complex Schwarzschild geometry defines an inverse constitutive
problem whose solution uniquely specifies the trace-sector constitutive
law within the minimal local constitutive framework developed in the
preceding sections.

%%%%%%%%%%%%%%%%%%%%%%%%%%%%%%%%%%%%%%%%%%%%%%%%%%%%%%%%%%%%%%

\subsection{Constitutive potential}

The constitutive equation
(\ref{eq:constitutiveODE}) is a first-order linear ordinary
differential equation whose solution is
\begin{equation}
W(\zeta) = W_0 \exp\!\left( \frac{n+2}{2}\zeta \right),
\label{eq:constitutivelaw} 
\end{equation}
where $W_0$ is an integration constant whose value depends on the
choice of reference configuration and is fixed, for that choice,
by matching to the known effective vacuum density.

The constitutive potential is therefore uniquely determined within the
minimal local constitutive framework. It depends exponentially on the
logarithmic transverse strain, with the exponent fixed entirely by the
effective vacuum structure of the static pseudo-complex Schwarzschild
solution.
Rather than being postulated as a constitutive ansatz, the exponential
form emerges as the solution of the inverse constitutive problem posed
by the effective vacuum sector.
This inversion parallels the reconstruction of a constitutive
potential from a known stress--strain relation in continuum mechanics
\cite{TruesdellNoll2004}.

%%%%%%%%%%%%%%%%%%%%%%%%%%%%%%%%%%%%%%%%%%%%%%%%%%%%%%%%%%%%%%

\subsection{Recovery of the static vacuum}

For a reference sphere of radius $r_0$, one may choose
\begin{equation}
\sqrt{q_0}=r_0^2\sin\theta ,
\end{equation}
so that
\begin{equation}
\zeta = -\ln\left(\frac{\sqrt q}{\sqrt{q_0}}\right) =
-2\ln\left(\frac{r}{r_0}\right).
\end{equation}
Substituting this expression into the constitutive potential
(\ref{eq:constitutivelaw}) gives
\begin{align}
W(r) &= W_0 \exp\!\left[ -(n+2) \ln\!\left( \frac{r}{r_0} \right)
  \right] \nonumber\\ &= W_0 \left( \frac{r_0}{r} \right)^{n+2}.
\end{align}
Hence,
\begin{equation}
W(r) \propto r^{-(n+2)}.
\end{equation}

The constitutive potential therefore reproduces exactly the radial
dependence of the effective vacuum energy density. Matching with
Eq.~(\ref{eq:rhoagain}) fixes the normalization,
\begin{equation}
W_0r_0^{\,n+2} = \frac{B_n} {8\pi G(n-2)},
\end{equation}
so that
\begin{equation}
W(r) = \frac{B_n} {8\pi G(n-2)} \frac{1}{r^{\,n+2}}.
\end{equation}

Thus the inverse constitutive construction reproduces the complete static
effective vacuum sector without introducing any additional
phenomenological parameters. The exponential constitutive potential
obtained from the inverse constitutive problem therefore provides an
exact constitutive representation of the known effective vacuum sector
of the static pseudo-complex Schwarzschild solution.

%%%%%%%%%%%%%%%%%%%%%%%%%%%%%%%%%%%%%%%%%%%%%%%%%%%%%%%%%%%%%%

\subsection{Interpretation}

The reconstruction carried out in this section establishes the
constitutive law of the static effective vacuum sector without
postulating a strain-energy function a priori. Instead, the known
pseudo-complex Schwarzschild solution determines the constitutive
potential through the inverse constitutive problem defined by the
effective vacuum tensor.

The reconstruction also reveals a useful geometric relation between
the radial vacuum profile and the constitutive law. The logarithmic
strain transforms the radial power-law behavior of the effective
vacuum density into a simple exponential constitutive law.

The reconstruction is unique only within the constitutive framework
developed in the preceding sections. It determines the constitutive
response of the isotropic trace sector but does not exclude more general
constitutive theories containing additional local or nonlocal invariants
that are not determined by the exact static spherically symmetric
solution.

The static pseudo-complex Schwarzschild geometry therefore provides the
minimal constitutive realization of the effective vacuum sector.
The following section examines how this constitutive framework extends
to the rotating pseudo-complex Kerr solution, where additional
transverse constitutive information becomes relevant.

%%%%%%%%%%%%%%%%%%%%%%%%%%%%%%%%%%%%%%%%%%%%%%%%%%%%%%%%%%%%%%%%%%%%%%%%%%
\section{Rotating Extension and the Shear Sector}
\label{sec:rotation}
%%%%%%%%%%%%%%%%%%%%%%%%%%%%%%%%%%%%%%%%%%%%%%%%%%%%%%%%%%%%%%%%%%%%%%%%%%

The constitutive framework developed in the preceding sections was
constructed from the exact static pseudo-complex Schwarzschild solution.
Spherical symmetry restricts the effective vacuum sector to the
isotropic trace mode of the transverse geometry, so that the
constitutive response is completely described by the logarithmic
transverse strain.

Rotation modifies this kinematic picture. Although the tensorial
structure of the effective vacuum sector is preserved, the reduction
to the static trace-sector constitutive law is no longer sufficient.
The purpose of the present section is therefore to identify the
additional transverse geometric information entering the rotating
pseudo-complex Kerr solution and to discuss its possible constitutive
interpretation in terms of the additional kinematic modes available
once spherical symmetry is broken.

%%%%%%%%%%%%%%%%%%%%%%%%%%%%%%%%%%%%%%%%%%%%%%%%%%%%%%%%%%%%%%

\subsection{Transverse stress structure}

The effective vacuum tensor of the rotating pseudo-complex Kerr solution
is most naturally analyzed in the Carter orthonormal frame. Its
principal stresses are
\begin{equation}
(-\rho,-\rho,p_\perp,p_\perp),
\label{eq:kerrEigen}
\end{equation}
so that the two transverse principal directions remain degenerate.
Consequently,
\begin{equation}
p_r=-\rho
\end{equation}
continues to hold, and the effective vacuum tensor retains the same
algebraic decomposition as in the static case,
\begin{equation}
T_{\mu\nu}^{\rm pc}
=
-\rho g_{\mu\nu}
+
(p_\perp-p_r)q_{\mu\nu}.
\end{equation}

The preservation of this tensorial structure is an important geometric
feature of the rotating solution. Rotation therefore does not alter the
fundamental decomposition into isotropic and transverse parts. Instead,
it modifies the constitutive response associated with the transverse
sector. As discussed below, this indicates that additional transverse
geometric information, absent from the static trace-sector description,
is required in the rotating case.

%%%%%%%%%%%%%%%%%%%%%%%%%%%%%%%%%%%%%%%%%%%%%%%%%%%%%%%%%%%%%%

\subsection{Geometric correction to the constitutive response}

For the rotating pseudo-complex Kerr solution, the ratio between the
transverse anisotropic stress and the effective vacuum energy density
becomes
\begin{equation}
\frac{p_\perp-p_r}{\rho} = \frac{n+2}{2} + \frac{n-2}{2}
\frac{a^2\cos^2\theta}{r^2}.
\label{eq:rotCorrection}
\end{equation}
The first term reproduces the static stress-to-density ratio of the
Schwarzschild solution, whereas the second introduces a purely
rotation-dependent correction.

The correction vanishes continuously in the nonrotating limit,
\begin{equation}
a \rightarrow 0,
\end{equation}
so that the static constitutive law is recovered exactly.

The additional dimensionless quantity
\begin{equation}
\eta = \frac{a^2\cos^2\theta}{r^2}
\end{equation}
provides a convenient measure of the rotational correction on the Kerr
background. It is not introduced as an independent constitutive
variable, but parametrizes the additional rotational dependence
of the transverse response.

%%%%%%%%%%%%%%%%%%%%%%%%%%%%%%%%%%%%%%%%%%%%%%%%%%%%%%%%%%%%%%

\subsection{Trace and shear}

Section~\ref{sec:transverse-geometry} established that the transverse
geometry possesses two irreducible kinematic modes: the trace mode,
which changes the transverse area, and the traceless shear mode, which
preserves the area to first order. The exact static pseudo-complex
Schwarzschild solution excites only the trace sector, leading to the
minimal constitutive theory developed in the preceding sections.

The rotating pseudo-complex Kerr solution indicates that the static
trace-sector description is no longer sufficient. Although the
effective vacuum tensor retains the same algebraic decomposition, the
rotational correction in Eq.~(\ref{eq:rotCorrection}) shows that additional transverse
geometric information enters the constitutive response.

Once spherical symmetry is broken, the transverse geometry admits, in
addition to the trace mode, the traceless shear sector identified in
Section~\ref{sec:transverse-geometry} and
\ref{app:representation}. Constitutive invariants constructed from
this sector therefore provide the natural candidates for extending the
minimal static theory. The known rotating stress tensor, however, does
not by itself uniquely determine such invariants or their constitutive
dependence.

A natural extension of the minimal constitutive potential may
therefore be written schematically as
\begin{equation}
W = W(\zeta,\mathcal{S}),
\end{equation}
where $\mathcal{S}$ denotes one or more scalar invariants constructed
from the transverse shear geometry. This expression is intended only
to indicate the enlarged constitutive configuration space. The present
analysis does not determine the functional dependence on $\mathcal{S}$.

%%%%%%%%%%%%%%%%%%%%%%%%%%%%%%%%%%%%%%%%%%%%%%%%%%%%%%%%%%%%%%

\subsection{Effective trace-sector response}

The rotating pseudo-complex Kerr solution does not determine the full
constitutive dependence of the potential on the transverse shear
invariants. Instead, it determines only the effective constitutive
response of the trace sector evaluated on the rotating background.
The rotating stress ratio suggests the following effective
trace-sector relation
\begin{equation}
\frac{\partial W}{\partial\zeta} = \left[ \frac{n+2}{2} +
  \frac{n-2}{2}\eta \right] W.
\end{equation}
This relation describes the rotational modification of
the trace response but does not determine the complete
constitutive dependence on the additional transverse
degrees of freedom.

A full constitutive theory of the rotating vacuum would require the
explicit identification of the independent shear invariants and their
associated constitutive response. The present analysis establishes only
the trace-sector response induced by the known rotating solution and
thereby identifies the geometric information required for a more general
constitutive description.

%%%%%%%%%%%%%%%%%%%%%%%%%%%%%%%%%%%%%%%%%%%%%%%%%%%%%%%%%%%%%%

\subsection{Constitutive interpretation}

The static and rotating solutions may therefore be viewed as probing
different sectors of the same transverse constitutive configuration space.
The static pseudo-complex Schwarzschild solution determines the constitutive
law along the trace direction, whereas the rotating solution shows that this
one-dimensional restriction is no longer sufficient.

The minimal theory developed here should accordingly be regarded as the
trace-sector restriction of a more general constitutive framework defined on
the full space of transverse metric deformations. Identifying the additional
invariants required by the rotating geometry and determining their
constitutive response remain open problems.

%%%%%%%%%%%%%%%%%%%%%%%%%%%%%%%%%%%%%%%%%%%%%%%%%%%%%%%%%%%%%%%%%%%%%%%%%%
\section{Constitutive Geometry and the Constraint Sector}
\label{sec:constraint-selection}
%%%%%%%%%%%%%%%%%%%%%%%%%%%%%%%%%%%%%%%%%%%%%%%%%%%%%%%%%%%%%%%%%%%%%%%%%%

Although inspired by classical hyperelastic constitutive theory
\cite{TruesdellNoll2004,Ogden1984,Holzapfel2000}, the present
construction remains entirely geometric and introduces no physical
material medium. Instead, it applies the mathematical structure of
hyperelastic constitutive theory to the effective vacuum sector of
pseudo-complex general relativity
\cite{HessGreiner2009,Caspar2012}.

The constitutive formulation developed in the preceding sections is
therefore complementary to the constrained variational formulation of
pseudo-complex gravity
\cite{Weber2026Variational}. The latter establishes the geometric
origin of the effective vacuum sector, whereas the present work
investigates a natural constitutive description of the effective
vacuum tensor generated by that geometric sector.

%%%%%%%%%%%%%%%%%%%%%%%%%%%%%%%%%%%%%%%%%%%%%%%%%%%%%%%%%%%%%%

\subsection{From projected stationarity to constitutive geometry}

The constrained variational formulation of pseudo-complex gravity shows
that the effective vacuum sector does not arise from independently
introduced matter fields. Instead, it emerges as the variational
response of a geometric constraint sector implementing the projected
pseudo-complex stationarity condition. In this sense, the constrained
variational principle explains the geometric origin of the effective
vacuum source.

The present work addresses a complementary question. Rather than asking
why an effective geometric source exists, we investigate whether its
tensorial structure admits a natural constitutive interpretation. The
analysis developed in the preceding sections demonstrates that the
effective vacuum sector is naturally described by a minimal local
constitutive theory defined on the transverse geometry.

The two approaches therefore address different levels of the theory.
The constrained variational formulation determines the geometric origin
of the effective vacuum sector, whereas the present constitutive
formulation provides a geometric organization of its effective stress
tensor. Taken together, these results suggest that the projected
effective vacuum sector may be represented as a constitutive response
of the geometric constraint structure associated with the
pseudo-complex variational principle.

%%%%%%%%%%%%%%%%%%%%%%%%%%%%%%%%%%%%%%%%%%%%%%%%%%%%%%%%%%%%%%

\subsection{The role of the transverse geometry}

The decomposition
\begin{equation}
T_{\mu\nu}^{\rm pc} = -\rho g_{\mu\nu} + (p_\perp-p_r)q_{\mu\nu}
\end{equation}
shows that the anisotropic contribution to the effective vacuum sector
acts exclusively on the transverse geometry. This observation suggests
that the relevant constitutive degrees of freedom are associated not
with arbitrary deformations of the four-dimensional spacetime metric but
with the intrinsic geometry of the transverse two-space.

The constitutive framework developed in the present work is based
precisely on this observation.
Rather than introducing an independent
effective medium or additional matter fields, it organizes the
constitutive response in terms of the irreducible kinematics of the
transverse geometry.
In particular, the logarithmic transverse strain $\zeta$ is not an
additional dynamical scalar field. It is constructed directly from the
intrinsic transverse metric and serves only as a constitutive coordinate
for its isotropic area-deformation mode.
The logarithmic transverse strain and its associated shear sector
provide the natural constitutive variables identified by this
geometric decomposition, in analogy with metric-based formulations
of two-dimensional elasticity
\cite{EfratiSharonKupferman2009,Yavari2026}.

From this perspective, the undetermined geometric constraint sector of
the projected pseudo-complex theory is not replaced but endowed with a
minimal constitutive description at the level of its effective vacuum
tensor. The transverse geometry therefore provides the natural
configuration space on which the constitutive formulation is built.

%%%%%%%%%%%%%%%%%%%%%%%%%%%%%%%%%%%%%%%%%%%%%%%%%%%%%%%%%%%%%%

\subsection{Minimal constitutive realization}

The constrained variational formulation determines the existence of an
effective geometric vacuum sector but does not uniquely specify its
constitutive description. The analysis developed in the present work
shows that, within the class of minimal local constitutive theories
introduced in Section~\ref{sec:minimal-action}, the logarithmic
transverse strain provides the natural constitutive variable governing
the static trace sector. Its minimality within this constitutive class
was established in \ref{app:minimality}.

The resulting constitutive action reproduces the complete static
effective vacuum sector of the pseudo-complex Schwarzschild solution.
Accordingly, the hyperelastic constitutive action should be interpreted
as a minimal local constitutive realization of the effective vacuum
sector generated by the geometric constraint structure. It is not a
derivation of the underlying pseudo-complex theory but a constitutive
description of its projected effective vacuum response.

%%%%%%%%%%%%%%%%%%%%%%%%%%%%%%%%%%%%%%%%%%%%%%%%%%%%%%%%%%%%%%

\subsection{Scope of the present work}

It is important to distinguish the results established here from
stronger claims that are beyond the scope of the present analysis.

The constitutive framework developed in this work does not derive
pseudo-complex gravity from elasticity, nor does it establish that the
logarithmic transverse strain follows uniquely from the
pseudo-complex variational principle itself. Instead, the analysis
shows that
\begin{enumerate}
\item
the effective vacuum sector admits a local constitutive description;
\item
within the present constitutive framework, the logarithmic
transverse strain is the natural constitutive variable governing the
static trace sector;
\item
the known static pseudo-complex Schwarzschild solution uniquely
determines the corresponding trace-sector constitutive potential.
\end{enumerate}
These results organize the effective vacuum sector within a systematic
constitutive framework while remaining fully consistent with the
underlying constrained variational formulation.

%%%%%%%%%%%%%%%%%%%%%%%%%%%%%%%%%%%%%%%%%%%%%%%%%%%%%%%%%%%%%%

\subsection{Perspective and selection principle}

Viewed together, the constrained variational formulation and the
constitutive framework developed in the present work suggest a natural
hierarchy.
The pseudo-complex variational principle gives rise to the projected
stationarity condition and thereby establishes the geometric origin of
the effective vacuum sector. The present work complements this
construction by providing a constitutive organization of the resulting
effective vacuum response through the transverse geometry and its
associated constitutive variables.
From this perspective, the constrained variational formulation and the
constitutive formulation address different levels of the theory. The
former explains the geometric origin of the effective vacuum sector,
whereas the latter provides a minimal local constitutive description of
its effective stress tensor. The two approaches are therefore
complementary rather than competing.

This viewpoint also suggests a possible criterion for selecting
admissible effective constraint sectors. Rather than allowing arbitrary
covariant tensors to represent the effective geometric source, one may
require that the resulting effective vacuum tensor admit a local
constitutive description on the transverse geometry. Such a
constitutive admissibility condition would substantially reduce the
freedom of the effective sector while remaining compatible with both
covariance and the projected pseudo-complex field equations. Whether
this criterion is sufficiently restrictive to identify a preferred
constraint sector remains an open question.

The constitutive framework developed here should therefore be regarded
as a minimal local organization of the effective vacuum sector generated
by the constrained pseudo-complex variational principle. It does not
replace the underlying variational theory but provides a systematic
geometric language for describing its projected effective vacuum
response.

%%%%%%%%%%%%%%%%%%%%%%%%%%%%%%%%%%%%%%%%%%%%%%%%%%%%%%%%%%%%%%%%%%%%%%%%%%
\section{Discussion and Outlook}
%%%%%%%%%%%%%%%%%%%%%%%%%%%%%%%%%%%%%%%%%%%%%%%%%%%%%%%%%%%%%%%%%%%%%%%%%%

The principal result of the present work is the demonstration that the
effective vacuum sector of pseudo-complex general relativity admits a
natural constitutive description based on the intrinsic geometry of the
transverse two-space. Rather than interpreting the projected effective
source as an anisotropic vacuum fluid, the analysis shows that it can be
organized as the constitutive response of the transverse geometry.

Within the class of minimal local constitutive theories considered here,
the logarithmic transverse strain provides the natural constitutive
variable governing the static trace sector. The corresponding
constitutive action, Eq.~(\ref{eq:constitutiveAction}), generates the constitutive
stress tensor, Eq.~(\ref{eq:constitutiveTensor}), which reproduces the complete
algebraic structure of the effective pseudo-complex vacuum tensor.
Matching this constitutive response to the known static
pseudo-complex Schwarzschild solution yields the inverse constitutive
equation, Eq.~(\ref{eq:constitutiveODE}), whose solution uniquely
determines the trace-sector constitutive potential,
Eq.~(\ref{eq:constitutivelaw}). The resulting constitutive law reproduces the
complete radial dependence of the effective vacuum density without
introducing additional phenomenological parameters.

The rotating pseudo-complex Kerr solution further supports the need for
a constitutive description extending beyond the static trace sector.
Although the tensorial decomposition of the effective vacuum sector
remains unchanged, rotation introduces additional transverse geometric
dependence that is not captured by the static trace-sector law. Within
the irreducible decomposition of the transverse geometry, the shear
sector provides the natural additional kinematic channel. Its complete
constitutive dependence, however, is not determined by the presently
known rotating solution.

It is important to distinguish the conclusions established here from
stronger claims that lie beyond the scope of the present work. The
analysis does not derive pseudo-complex gravity from elasticity, nor
does it establish that the logarithmic transverse strain follows
uniquely from the pseudo-complex variational principle itself. Instead,
it demonstrates that the effective vacuum sector admits a systematic
constitutive description that is fully consistent with the constrained
variational formulation.

Several important questions remain open.
\begin{enumerate}
\item Can the logarithmic transverse strain be obtained directly from
the underlying pseudo-complex variational principle rather than from
the effective vacuum sector?

\item What additional transverse constitutive invariants are required
for the rotating pseudo-complex Kerr solution, and to what extent can
their constitutive dependence be reconstructed from the known rotating
vacuum sector?

\item Does constitutive admissibility provide a useful criterion for
selecting physically preferred geometric constraint sectors within the
general constrained variational formulation?

\item Can the constitutive framework be extended beyond stationary
spacetimes to dynamical solutions, gravitational waves, and
cosmological models?
\end{enumerate}

More broadly, the present work suggests that constitutive geometry may
provide a natural mathematical language for organizing the effective
vacuum sector generated by the constrained pseudo-complex variational
principle. Whether this constitutive structure reflects a deeper
property of the underlying pseudo-complex geometry or should instead be
regarded as an effective description remains an open question.
Nevertheless, the existence of a minimal local constitutive realization
demonstrates that the effective vacuum sector possesses a previously
unrecognized geometric organization extending beyond its traditional
interpretation as an anisotropic vacuum fluid.

%%%%%%%%%%%%%%%%%%%%%%%%%%%%%%%%%%%%%%%%%%%%%%%%%%%%%%%%%%%%%%%%%%%%%%%%%%%%%%%%%

%\backmatter

%\section*{Acknowledgements}
%P.O.H. acknowledges financial support from PAPIIT-DGAPA (IN116824).

\section*{Declarations}

\begin{itemize}

\item \textbf{Conflict of interest/Competing interests} \\
The authors declare that they have no competing interests.

\item \textbf{Ethics approval and consent to participate} \\
Not applicable.

\item \textbf{Consent for publication} \\
Not applicable.

\item \textbf{Data availability} \\
No datasets were generated or analyzed during the current study.

\item \textbf{Materials availability} \\
Not applicable.

\item \textbf{Code availability} \\
No code was used or generated in this study.

\item \textbf{Author contribution} \\
All authors contributed equally to the conception, analysis, and writing of this work.

\end{itemize}

\bigskip

%%%%%%%%%%%%%%%%%%%%%%%%%%%%%%%%%%%%%%%%%%%%%%%%%%%%%%%%%%%%%%%%%%%%%%%%%%%%%%%%%

%%%%%%%%%%%%%%%%%%%%%%%%%%%%%%%%%%%%%%%%%%%%%%%%%%%%%%%%%%%%%%%%%%%%%%%%%%
\appendix
%%%%%%%%%%%%%%%%%%%%%%%%%%%%%%%%%%%%%%%%%%%%%%%%%%%%%%%%%%%%%%%%%%%%%%%%%%

%%%%%%%%%%%%%%%%%%%%%%%%%%%%%%%%%%%%%%%%%%%%%%%%%%%%%%%%%%%%%%%%%%%%%%%%%%
\section{Variations of the Transverse Geometry}
\label{app:geometry}
%%%%%%%%%%%%%%%%%%%%%%%%%%%%%%%%%%%%%%%%%%%%%%%%%%%%%%%%%%%%%%%%%%%%%%%%%%

This appendix collects the geometric identities underlying the
constitutive formulation developed in the main text. In particular, it
establishes the relation between variations of the transverse metric,
the corresponding transverse area element, and the logarithmic
transverse strain used throughout the constitutive construction.

%%%%%%%%%%%%%%%%%%%%%%%%%%%%%%%%%%%%%%%%%%%%%%%%%%%%%%%%%%%%%%

\subsection{Variation of the transverse determinant}

Let $q_{AB}$ denote the induced metric on the two-dimensional
transverse manifold, with determinant $q=\det(q_{AB})$.
The inverse metric satisfies
\begin{equation}
q^{AC}q_{CB} = \delta^A_{\;B}.
\end{equation}
The variation of the determinant follows from Jacobi's identity,
\begin{equation}
\delta q = q\,q^{AB}\delta q_{AB}.  
\label{eq:JacobiIdentity}
\end{equation}
Dividing by \(q\) immediately gives
\begin{equation}
 \delta\ln q = q^{AB}\delta q_{AB}.  
\label{eq:dlogq}
\end{equation}
Since
\begin{equation}
\ln\sqrt q = \frac12\ln q,
\end{equation}
one obtains
\begin{equation}
 \delta\ln\sqrt q = \frac12 q^{AB}\delta q_{AB}.  
\label{eq:dlogsqrtqApp}
\end{equation}

Equation (\ref{eq:dlogsqrtqApp}) shows that the trace of the
transverse metric variation is precisely the fractional variation of
the local transverse area element.

%%%%%%%%%%%%%%%%%%%%%%%%%%%%%%%%%%%%%%%%%%%%%%%%%%%%%%%%%%%%%%

\subsection{Logarithmic transverse strain}

Introduce the logarithmic transverse area strain
\begin{equation}
\zeta = - \ln \left( \frac{\sqrt q} {\sqrt{q_0}} \right),
\end{equation}
where \(q_0\) denotes a fixed reference configuration.
Since
\begin{equation}
\delta q_0=0,
\end{equation}
Eq.~(\ref{eq:dlogsqrtqApp}) immediately yields
\begin{equation}
 \delta\zeta = - \frac12 q^{AB}\delta q_{AB}.  
\label{eq:deltazetaAppendix}
\end{equation}

The logarithmic transverse strain is therefore conjugate to the trace
deformation of the transverse metric and provides the natural scalar
coordinate along the isotropic deformation direction.

%%%%%%%%%%%%%%%%%%%%%%%%%%%%%%%%%%%%%%%%%%%%%%%%%%%%%%%%%%%%%%

\subsection{Covariant and contravariant metric variations}

The inverse transverse metric satisfies
\begin{equation}
q^{AC}q_{CB} = \delta^A_{\;B}.
\end{equation}
Taking the variation gives
\begin{equation}
(\delta q^{AC})q_{CB} + q^{AC}\delta q_{CB} = 0,
\end{equation}
which is  equivalent to
\begin{equation}
 \delta q_{AB} = - q_{AC}q_{BD} \delta q^{CD}.  
\label{eq:inversevariation}
\end{equation}
Contracting with \(q^{AB}\) gives
\begin{equation}
q^{AB}\delta q_{AB} = - q_{AB}\delta q^{AB},
\end{equation}
and therefore
\begin{equation}
\delta\zeta = \frac12 q_{AB} \delta q^{AB}.  
\label{eq:contravariantvariation}
\end{equation}

%%%%%%%%%%%%%%%%%%%%%%%%%%%%%%%%%%%%%%%%%%%%%%%%%%%%%%%%%%%%%%

\subsection{Embedding into spacetime}

The transverse metric is the pullback of the spacetime projector
\begin{equation}
q_{\mu\nu} = g_{\mu\nu} + u_\mu u_\nu - s_\mu s_\nu .
\end{equation}
Within the restricted constitutive variation considered in the present
work, the timelike and radial directions defining the transverse
splitting are held fixed, and the metric variation is restricted to the
induced transverse geometry. Under this restricted variation, the
pullback of the spacetime metric to the transverse manifold satisfies
\begin{equation}
\delta q_{AB}=\delta g_{AB},
\end{equation}
and hence
\begin{equation}
  \delta\zeta = \frac{1}{2}q_{\mu\nu}\delta g^{\mu\nu}.
  \label{eq:spacetimevariationAppendix}
\end{equation}
This is the restricted metric variation used in Section~\ref{sec:minimal-action} to define
the minimal transverse constitutive response.

%%%%%%%%%%%%%%%%%%%%%%%%%%%%%%%%%%%%%%%%%%%%%%%%%%%%%%%%%%%%%%

\subsection{Geometric interpretation}

Equations~(\ref{eq:dlogsqrtqApp}) and
(\ref{eq:spacetimevariationAppendix}) show that the logarithmic
transverse strain measures infinitesimal changes of the local
transverse area. Indeed,
\begin{equation}
\delta\sqrt q
=
\sqrt q\,\delta\ln\sqrt q
=
-\sqrt q\,\delta\zeta.
\end{equation}
The logarithmic transverse strain therefore provides the natural scalar
coordinate along the isotropic direction in the space of transverse
metrics. It measures infinitesimal changes of the transverse area while
remaining insensitive to traceless shear deformations.

By contrast, traceless transverse deformations satisfy
\begin{equation}
q^{AB}\delta q^{\rm TF}_{AB}=0,
\end{equation}
and therefore leave both the transverse area and the logarithmic
strain unchanged to first order.

The corresponding irreducible decomposition of transverse metric
variations into trace and shear sectors is developed in
\ref{app:representation}.

%%%%%%%%%%%%%%%%%%%%%%%%%%%%%%%%%%%%%%%%%%%%%%%%%%%%%%%%%%%%%%%%%%%%%%%%%%
\section{Representation-Theoretic Structure of the Transverse Geometry}
\label{app:representation}
%%%%%%%%%%%%%%%%%%%%%%%%%%%%%%%%%%%%%%%%%%%%%%%%%%%%%%%%%%%%%%%%%%%%%%%%%%

The constitutive formulation developed in the main text relies on the
fact that infinitesimal deformations of the transverse geometry admit a
natural decomposition into irreducible components under the local
rotation group acting on the transverse tangent space. The purpose of
this appendix is to derive this decomposition and to identify the
corresponding trace and shear sectors.

%%%%%%%%%%%%%%%%%%%%%%%%%%%%%%%%%%%%%%%%%%%%%%%%%%%%%%%%%%%%%%

\subsection{Space of metric deformations}

Let $q_{AB}$ denote the metric of the two-dimensional transverse
manifold. An arbitrary infinitesimal deformation is represented by the
symmetric tensor
\begin{equation}
\delta q_{AB}
=
\delta q_{BA}.
\end{equation}
Since the transverse manifold is two-dimensional,
\[
A,B=1,2,
\]
and the space of symmetric rank-two tensors possesses three independent
components.

%%%%%%%%%%%%%%%%%%%%%%%%%%%%%%%%%%%%%%%%%%%%%%%%%%%%%%%%%%%%%%%%%%%%%%%%%%

\subsection{Trace decomposition}

Every symmetric rank-two tensor on a two-dimensional Riemannian
manifold admits the unique decomposition into trace and traceless
parts,
\begin{equation}
\delta q_{AB}
=
\delta q^{\rm tr}_{AB}
+
\delta q^{\rm TF}_{AB},
\label{eq:traceSplitAppendix}
\end{equation}
where
\begin{equation}
\delta q^{\rm tr}_{AB}
=
\frac12 q_{AB}q^{CD}\delta q_{CD},
\end{equation}
and
\begin{equation}
\delta q^{\rm TF}_{AB}
=
\delta q_{AB}
-
\frac12 q_{AB}q^{CD}\delta q_{CD}.
\end{equation}

The traceless property follows immediately by contraction with the
inverse metric:
\begin{align}
q^{AB}\delta q^{\rm TF}_{AB}
&=
q^{AB}\delta q_{AB}
-
\frac12
q^{AB}q_{AB}
q^{CD}\delta q_{CD}
\nonumber\\
&=
q^{AB}\delta q_{AB}
-
q^{CD}\delta q_{CD}
\nonumber\\
&=
0,
\end{align}
where
\[
q^{AB}q_{AB}=2
\]
has been used.

The trace and traceless tensors are therefore orthogonal under
contraction with the transverse metric and provide the unique
irreducible decomposition of infinitesimal transverse metric
deformations.

%%%%%%%%%%%%%%%%%%%%%%%%%%%%%%%%%%%%%%%%%%%%%%%%%%%%%%%%%%%%%%

\subsection{Uniqueness}

To establish uniqueness, suppose that
\[
\delta q_{AB}
=
X_{AB}
+
Y_{AB},
\]
where
\[
X_{AB}
=
\alpha q_{AB},
\]
is purely trace, while
\[
q^{AB}Y_{AB}=0
\]
is traceless.

Taking the trace of the decomposition gives
\[
2\alpha
=
q^{AB}\delta q_{AB},
\]
and therefore
\[
\alpha
=
\frac12 q^{AB}\delta q_{AB}.
\]
Substituting this expression back into the decomposition yields
\[
X_{AB}
=
\frac12 q_{AB}q^{CD}\delta q_{CD},
\]
and
\[
Y_{AB}
=
\delta q_{AB}
-
\frac12 q_{AB}q^{CD}\delta q_{CD}.
\]

Both components are therefore uniquely determined by
$\delta q_{AB}$, proving that the decomposition
(\ref{eq:traceSplitAppendix}) is unique.

%%%%%%%%%%%%%%%%%%%%%%%%%%%%%%%%%%%%%%%%%%%%%%%%%%%%%%%%%%%%%%

\subsection{Representation-theoretic interpretation}

At each point of the transverse two-dimensional manifold, the tangent
space carries the natural action of the local rotation group
\[
SO(2).
\]
The space of symmetric rank-two tensors therefore decomposes into the
irreducible representations
\begin{equation}
\mathrm{Sym}^2(T^*\Sigma)
=
\mathbf{1}
\oplus
\mathbf{2},
\end{equation}
where $\mathbf{1}$ denotes the one-dimensional trace representation and
$\mathbf{2}$ the two-dimensional traceless symmetric representation.

The trace representation corresponds to isotropic changes of the
transverse area, whereas the traceless representation describes shear
deformations that preserve the area to first order. This decomposition
is a purely geometric property of symmetric tensors on a
two-dimensional Riemannian manifold and is therefore independent of the
pseudo-complex field equations.

The constitutive formulation developed in the main text exploits this
representation-theoretic structure by identifying the static
pseudo-complex Schwarzschild solution with the trace sector. Once
spherical symmetry is broken, as in the rotating pseudo-complex Kerr
solution, the traceless shear representation provides the natural
additional kinematic sector available to the transverse geometry.

%%%%%%%%%%%%%%%%%%%%%%%%%%%%%%%%%%%%%%%%%%%%%%%%%%%%%%%%%%%%%%

\subsection{Area and shear}

The representation-theoretic decomposition has a direct geometric
interpretation. Since
\begin{equation}
\delta\ln\sqrt q = \frac12 q^{AB}\delta q_{AB},
\end{equation}
only the trace component contributes to the variation of the transverse
area element. By contrast, the traceless component satisfies
\begin{equation}
q^{AB}\delta q^{\rm TF}_{AB}=0,
\end{equation}
and therefore preserves the transverse area to first order.

Accordingly, the two irreducible sectors have distinct geometric roles:
the trace mode describes isotropic changes of the local transverse area,
whereas the traceless mode describes area-preserving shear
deformations.

%%%%%%%%%%%%%%%%%%%%%%%%%%%%%%%%%%%%%%%%%%%%%%%%%%%%%%%%%%%%%%

\subsection{Constitutive significance}

The decomposition derived above provides the geometric foundation of the
constitutive formulation developed in the main text. For the exact static
pseudo-complex Schwarzschild solution, spherical symmetry restricts the
constitutive response to the trace representation, while the traceless
representation supplies the additional kinematic sector available once
spherical symmetry is broken.

The representation-theoretic decomposition therefore provides the natural
kinematic framework for extending the minimal trace-sector constitutive
theory to more general pseudo-complex spacetimes.

%%%%%%%%%%%%%%%%%%%%%%%%%%%%%%%%%%%%%%%%%%%%%%%%%%%%%%%%%%%%%%%%%%%%%%%%%%
\section{Minimality of the Logarithmic Transverse Strain}
\label{app:minimality}
%%%%%%%%%%%%%%%%%%%%%%%%%%%%%%%%%%%%%%%%%%%%%%%%%%%%%%%%%%%%%%%%%%%%%%%%%%

This appendix establishes that the logarithmic transverse area strain
introduced in the main text is the natural scalar constitutive coordinate
within the class of local constitutive theories considered in this work.
The result should not be interpreted as a uniqueness theorem for the
full pseudo-complex theory. Rather, it identifies the minimal local
constitutive scalar associated with isotropic transverse deformation.

%%%%%%%%%%%%%%%%%%%%%%%%%%%%%%%%%%%%%%%%%%%%%%%%%%%%%%%%%%%%%%

\subsection{Constitutive assumptions}

The proof is restricted to scalar constitutive variables satisfying the
following assumptions:
\begin{enumerate}
\item They depend locally on the transverse metric,
\[
\psi=\psi(q_{AB}).
\]

\item They depend only algebraically on the transverse metric and
contain no derivatives of $q_{AB}$.

\item Their first variation is proportional to the isotropic trace
deformation with a state-independent normalization,
\begin{equation}
  \delta\psi = C\,q^{AB}\delta q_{AB},
  \label{eq:tracecondition}
\end{equation}
where $C$ is a constant.
\end{enumerate}

The third assumption is stronger than requiring only that the variation
be purely trace. It selects a constitutive coordinate whose differential
is directly proportional to the transverse trace projector with a
state-independent normalization.

%%%%%%%%%%%%%%%%%%%%%%%%%%%%%%%%%%%%%%%%%%%%%%%%%%%%%%%%%%%%%%

\subsection{Minimality theorem}

\begin{theorem}[Minimality of the logarithmic transverse area strain]
Let $\psi(q_{AB})$ satisfy the constitutive assumptions stated above.
Then
\begin{equation}
\psi = 2C\ln\sqrt{q} + \mathrm{const.},
\end{equation}
where
\[
q=\det(q_{AB}).
\]
Thus, up to an overall normalization and an additive constant, the
logarithmic transverse area is the unique local constitutive coordinate
whose first variation is proportional to the transverse trace
deformation with a state-independent proportionality factor.
\end{theorem}

\begin{proof}
By locality,
\[
\psi=\psi(q_{AB}),
\]
and therefore
\begin{equation}
\delta\psi = \frac{\partial\psi}{\partial q_{AB}} \delta q_{AB}.
\end{equation}
Comparison with Eq.~(\ref{eq:tracecondition}) gives
\begin{equation}
\frac{\partial\psi}{\partial q_{AB}} = Cq^{AB}.
\label{eq:partialcondition}
\end{equation}
Using Jacobi's identity,
\begin{equation}
\delta\ln q = q^{AB}\delta q_{AB},
\end{equation}
it follows that
\begin{equation}
\delta\psi = C\,\delta\ln q.
\end{equation}
Integration therefore yields
\begin{equation}
\psi = C\ln q +\mathrm{const.}  = 2C\ln\sqrt q +\mathrm{const.},
\end{equation}
which proves the theorem.
\end{proof}

%%%%%%%%%%%%%%%%%%%%%%%%%%%%%%%%%%%%%%%%%%%%%%%%%%%%%%%%%%%%%%

\subsection{Relation to the logarithmic strain}

Introducing the reference configuration \(q_0\),
\begin{equation}
\zeta = -\ln\left(\frac{\sqrt q}{\sqrt{q_0}}\right),
\end{equation}
one finds
\begin{equation}
\delta\zeta = -\frac12 q^{AB}\delta q_{AB},
\end{equation}
which is precisely the constitutive variation required in the main
text.

The logarithmic transverse strain therefore realizes the minimal
constitutive scalar identified by the theorem. Relative to the general
solution, it differs only by the conventional normalization
\(C=-\tfrac12\) and by the choice of the reference configuration
\(q_0\), which fixes the additive integration constant.

Consequently, the logarithmic transverse strain is not an ad hoc
constitutive choice but the natural local scalar associated with
isotropic transverse deformation within the class of constitutive
theories considered in this work.

The scope of the theorem should also be emphasized. It applies only to
local constitutive scalars that depend algebraically on the transverse
metric and whose first variation is proportional to the isotropic
trace deformation with a state-independent normalization.

The theorem does not exclude additional constitutive invariants
constructed from transverse shear, intrinsic curvature, or derivatives
of the transverse metric. Such quantities are not fixed by the static
trace-sector matching considered in the present work and may become
relevant in more general constitutive descriptions of rotating or
dynamical pseudo-complex spacetimes.

Accordingly, the logarithmic transverse strain should be viewed as the
minimal constitutive variable governing the static trace sector rather
than as the unique scalar characterizing the complete constitutive
geometry.

%%%%%%%%%%%%%%%%%%%%%%%%%%%%%%%%%%%%%%%%%%%%%%%%%%%%%%%%%%%%%%

%\funding{P.O.H. acknowledges financial support from PAPIIT-DGAPA (IN116824).}

\bibliographystyle{spphys}
\bibliography{references}

\end{document}